\documentclass[DIV12,11pt]{scrartcl}

\usepackage{amssymb,amsmath}
\usepackage{cite}
\usepackage{authblk}
\usepackage{amsthm}

\newcommand{\bigmid}{~\big|~}

\setkomafont{title}{\normalfont \LARGE}
\setkomafont{section}{\normalfont \Large \bfseries \boldmath}
\setkomafont{subsection}{\normalfont \large \bfseries \boldmath}
\setkomafont{subsubsection}{\normalfont \normalsize \bfseries \boldmath}
\setkomafont{descriptionlabel}{\bfseries\boldmath}
\setkomafont{paragraph}{\normalfont \bfseries \boldmath}

\DeclareMathOperator{\probab}{\mathbb{P}}
\DeclareMathOperator{\expected}{\mathbb{E}}
\DeclareMathOperator{\Exp}{\operatorname{Exp}}

\DeclareMathOperator{\trivial}{\mathsf{TRIVIAL}}
\DeclareMathOperator{\kmedian}{\mathsf{MEDIAN}}
\DeclareMathOperator{\greedy}{\mathsf{GREEDY}}
\DeclareMathOperator{\nn}{\mathsf{NN}}
\DeclareMathOperator{\instsp}{\mathsf{INSERT}}

\DeclareMathOperator{\mm}{\mathsf{MM}}
\DeclareMathOperator{\tsp}{\mathsf{TSP}}
\DeclareMathOperator{\cost}{\mathsf{cost}}

\newcommand{\nat}{\mathbb{N}}

\newcommand{\eps}{\varepsilon}
\newcommand{\dmax}{\Delta_{\operatorname{max}}}

\newtheorem{theorem}{Theorem}[section]
\newtheorem{lemma}[theorem]{Lemma}
\newtheorem{corollary}[theorem]{Corollary}
\newtheorem{claim}[theorem]{Claim}

\newcommand{\deltamin}{\delta_{\min}}

\title{Random Shortest Paths: Non-Euclidean Instances for Metric Optimization Problems\thanks{To appear in \emph{Algorithmica}. An extended abstract of this work
has appeared in the \emph{Proceedings of the 38th Int.\ Symp.\ on Mathematical Foundations of Computer Science (MFCS 2013)}.}}

\author[1]{Karl Bringmann\thanks{Karl Bringmann is a recipient of the \emph{Google Europe Fellowship in Randomized Algorithms}, and
   this research is supported in part by this Google Fellowship.}}
\author[2]{Christian Engels}
\author[3]{Bodo Manthey}
\author[4]{B.~V.~Raghavendra~Rao}

\affil[1]{Max Planck Institute for Informatics, Saarbr\"ucken, Germany, \texttt{kbringma@mpi-inf.mpg.de}}
\affil[2]{Saarland University, Saarbr\"ucken, Germany, \texttt{engels@cs.uni-saarland.de}}
\affil[3]{University of Twente, Enschede, Netherlands, \texttt{b.manthey@utwente.nl}}
\affil[4]{Indian Institute of Technology Madras, Chennai, India, \texttt{bvrr@cse.iitm.ac.in}}

\date{May 23, 2014}

\begin{document}

\maketitle

\begin{abstract}
Probabilistic analysis for metric optimization problems has mostly been conducted on random Euclidean instances, but little
is known about metric instances drawn from distributions other than the Eu\-clid\-e\-an.
This motivates our study of random metric instances for optimization problems obtained as follows:
Every edge of a complete graph gets a weight drawn independently at random. The distance between two nodes is then the length of a shortest path 
(with respect to the weights drawn) that connects these nodes.

We prove structural properties of the random shortest path metrics generated in this way. Our main structural contribution is the construction of a good clustering.
Then we apply these findings to analyze the approximation ratios
of heuristics for matching, the traveling salesman problem (TSP), and the $k$-median problem,
as well as the running-time of the 2-opt heuristic for the TSP.
The bounds that we obtain are considerably better than
the respective worst-case bounds. 
This suggests that random shortest path metrics are easy instances, similar to random Euclidean instances, albeit for completely different structural reasons.
\end{abstract}

\section{Introduction}
\label{sec:introduction}

For large-scale optimization problems, finding optimal solutions within
reasonable time is often impossible, because many such problems,
like the traveling salesman
problem (TSP), are NP-hard.
Nevertheless, we often observe that simple heuristics succeed surprisingly
quickly in finding close-to-optimal solutions.
Many such heuristics perform well in practice but have a poor worst-case performance.
In order to  explain the performance of such heuristics,
probabilistic analysis has proved to be a useful alternative to worst-case analysis.
Probabilistic analysis of optimization problems has
been conducted with respect to arbitrary instances (without the triangle
inequality)~\cite{Frieze04,Karp77}
or instances embedded in Euclidean space.
In particular, the limiting behavior of various heuristics for many of the
Euclidean optimization
problems is known precisely~\cite{Yukich1998}.

However, the average-case performance of heuristics for general metric instances is not well
understood. This lack of understanding can be explained by two reasons:
First, independent random edge lengths (without the triangle inequality) and random geometric instances are
relatively easy to handle from a technical point of view -- the former because of the independence of the lengths,
the latter because Euclidean space provides a structure that can be exploited.
Second, analyzing heuristics on random metric spaces requires an understanding of random
metric spaces in the first place. While Vershik~\cite{Ver04} gave an analysis
of a process for obtaining random metric spaces, using this directly to analyze algorithms seems difficult.

In order to initiate systematic research of heuristics on general metric spaces,
we use the following model, proposed by Karp and Steele~\cite[Section 3.4]{KarpSteele}:
given an undirected complete graph, we draw edge weights independently at random according to
exponential distributions with parameter one.
The distance between any two vertices is the total weight of the shortest path between them, measured with respect to the random weights.
We call such instances \emph{random shortest path metrics}.

This model is also known as \emph{first-passage percolation},
and has been introduced by Broadbent and Hammersley
as a model for passage of fluid in a porous
medium~\cite{BS10,BH57}.
More recently, it has also been used to model shortest paths in networks such as the Internet~\cite{EGRHN12}. 
The appealing feature of random shortest path metrics is their simplicity, which enables
us to use them for the analysis of heuristics.

\subsection{Known and Related Results}

There has been significant study of  random shortest path metrics or first-passage percolation.  The 
expected length of an edge is known to be 
$\ln n/n$~\cite{DavisPrieditis:ExpectedShortestPath:1993,Janson}.
Asymptotically the same bound holds
also for the longest edge almost surely~\cite{HassinZemel:ShortestPath:1985,Janson}.
These results hold not only for the exponential distribution, but for every distribution with distribution function $F$ satisfying
$F(x) = x + o(x)$ for small values of $x$~\cite{Janson}. (See also Section~\ref{sec:open}.)
This model has been used to analyze algorithms for computing shortest paths~\cite{HassinZemel:ShortestPath:1985,PeresEA,FriezeGrimmett}.
Kulkarni and Adlakha have developed algorithmic methods to compute distribution and moments of several
optimization problems~\cite{Kulkarni:MST:1988,KulkarniAdlakha:FlowExponential:1985,Kulkarni:SP:1986}.
Beyond shortest path algorithms, random shortest path metrics have been applied only rarely to analyze algorithms.
Dyer and Frieze~\cite{DyerFrieze}, answering a question raised by Karp and Steele~\cite[Section 3.4]{KarpSteele},
analyzed the patching heuristic for the asymmetric TSP (ATSP) in this model.
They showed that it comes within a factor of $1+o(1)$ of the optimal solution with
high probability.
Hassin and Zemel~\cite{HassinZemel:ShortestPath:1985} applied their findings to the $1$-center problem.

From a more structural point of view, first-passage percolation
has been analyzed in the area of complex networks, where
the hop-count (the number
of edges on a shortest path) and the length of shortest path
trees have been analyzed~\cite{Hofstad06}.
These properties have also been studied on random graphs with random edge
weights in various settings~\cite{HofstadEA:FirstPassageRandom:2001,ErdosFPP,FPPInhomogeneous,Universality,BhamidiEA:FirstPassageFinite:2010}.
Addario-Berry et al.~\cite{BBL10} have shown that the number of edges in the longest of the
shortest paths is $O(\log n)$ with high probability, and hence the shortest path trees have depth $O(\log n)$.

\subsection{Our Results}
As far as we are aware, simple
heuristics such as greedy heuristics have
not been studied in this model yet. Understanding the performance of such algorithms is particularly important as they are easy to implement and
used in many applications.

We provide a probabilistic analysis of simple heuristics for optimization
under random shortest path metrics.
First, we provide structural properties of 
random shortest path metrics (Section~\ref{sec:structure}). Our most important structural contribution is proving the existence 
of a good clustering (Lemma~\ref{lem:numberofclusters}).
Then we use these structural insights to analyze simple algorithms for 
minimum weight matching and the TSP to obtain better expected approximation ratios compared to the worst-case bounds.
In particular,  we show that the greedy algorithm for
minimum-weight perfect matching (Theorem~\ref{thm:greedy-exp}), the nearest-neighbor
heuristic for the TSP (Theorem~\ref{thm:nn-const}), and every insertion heuristic for the TSP (Theorem~\ref{thm:tsp-ins})
achieve constant expected approximation ratios. We also analyze the 
2-opt heuristic for the TSP
and show that the expected number of
2-exchanges required before the termination of the algorithm  is
bounded by $O(n^8\log^3 n)$ (Theorem~\ref{thm:2-opt}).
Investigating further the structural properties of random shortest path metrics, we then consider the $k$-median problem
(Section~\ref{sec:kmedian}), and show that the most trivial procedure
of choosing $k$ arbitrary vertices as $k$-median yields a $1+o(1)$ approximation in 
expectation, provided $k=O(n^{1-\eps})$ for some $\eps>0$ (Theorem~\ref{thm:kcentre}).

\section{Model and Notation}

We consider undirected complete graphs $G=(V,E)$ without loops.
First, we draw \emph{edge weights} $w(e)$ independently at random
according to the exponential distribution\footnote{Exponential distributions
are technically the easiest to handle because they are memoryless.
We will discuss other distributions in Section~\ref{sec:open}.} with parameter $1$.

Second, let the \emph{distances} $d: V \times V \to [0, \infty)$ be given as follows:
the distance $d(u,v)$ between $u$ and $v$ is the minimum total weight of a path
connecting $u$ and $v$.
In particular, we have $d(v,v) = 0$ for all $v \in V$,
$d(u,v) = d(v,u)$ because $G$ is undirected,
and the triangle inequality: $d(u,v) \leq d(u,x) + d(x,v)$
for all $u, x, v \in V$.
We call the complete graph with distances $d$ obtained from random weights $w$ a \emph{random shortest path metric}.

We use the following notation:
Let $\dmax = \max_{u,v} d(u,v)$ denote the \emph{diameter} of the random shortest path metric.
Let $B_\Delta(v) = \{u \in V \mid d(u,v) \leq \Delta\}$ be the
ball of radius $\Delta$ around $v$, i.e., the set of all nodes whose distance to $v$ is at most $\Delta$.

We denote the
minimal $\Delta$ such that there are at least $k$ nodes within a distance of $\Delta$ of $v$ by $\tau_{k}(v)$. Formally,
we define $\tau_k(v) = \min \{\Delta \mid |B_{\Delta}(v)| \geq k\}$.

By $\Exp(\lambda)$, we denote the exponential distribution with parameter $\lambda$. If a random variable~$X$ is distributed
according to a probability distribution $P$, we write $X \sim P$. In particular, $X \sim 
\sum_{i=1}^m \Exp(\lambda_i)$ means that $X$ is the sum of $m$ independent exponentially distributed random variables with
parameters $\lambda_1, \ldots, \lambda_m$.

By $\exp$, we denote the exponential function.
For $n \in \nat$, let $[n] = \{1,\dots,n\}$
and let $H_n = \sum_{i=1}^n 1/i$ be the $n$-th harmonic number. 

\section{Structural Properties of Shortest Path Metrics}
\label{sec:structure}

\subsection{Random Process}

To understand random shortest path metrics, it is convenient to fix a starting
vertex $v$ and see how the lengths from $v$ to the other vertices develop.
In this way, we analyze the distribution of~$\tau_k(v)$.

The values $\tau_k(v)$ are generated by a simple birth process as follows.
(The same process has been analyzed by Davis and Prieditis~\cite{DavisPrieditis:ExpectedShortestPath:1993},
Janson~\cite{Janson}, and also in subsequent papers.)
For $k=1$, we have $\tau_k(v) = 0$.

For $k \geq 1$, we are looking for the closest vertex to any
vertex in $B_{\tau_{k}(v)}(v)$ in order to obtain $\tau_{k+1}(v)$.
This conditions all edges $(u, x)$ with $u \in B_{\tau_{k}(v)}(v)$
and $x \notin B_{\tau_{k}(v)}(v)$ to be of length at least
$\tau_{k}(v) - d(v,u)$.
The set $B_{\tau_{k}(v)}(v)$ contains $k$ vertices. Thus, there are $k \cdot (n-k)$
edges to the rest of the graph.
Consequently, the difference $\delta_k = \tau_{k+1}(v) - \tau_{k}(v)$ is distributed as the minimum of $k(n-k)$ exponential random variables
(with parameter 1), or, equivalently, as $\Exp(k \cdot (n-k))$.
We obtain that $\tau_{k+1}(v) \sim \sum_{i=1}^{k} \Exp\bigl(i \cdot (n-i)\bigr)$.
Note that these exponential distributions as well as the random variables $\delta_1,
\ldots, \delta_n$ are independent.

Exploiting linearity of expectation and that the expected value of $\Exp(\lambda)$ is $1/\lambda$ we obtain the following lemma.

\begin{lemma}
\label{thm:expectations}
For any $k \in [n]$ and any $v \in V$, we have
$\expected\bigl(\tau_k(v) \bigr) = \frac 1n \cdot \bigl(H_{k-1} + H_{n-1} - H_{n-k}\bigr)$
and $\tau_k(v) \sim \sum_{i=1}^{k-1} \Exp\bigl(i \cdot (n-i)\bigr)$.
\end{lemma}
\begin{proof}
The proof is by induction on $k$. For $k=1$, we have $\tau_k(v) = 0$ and
$H_{k-1} + H_{n-1} - H_{n-k} = H_{0} + H_{n-1} - H_{n-1} = 0$.
Now assume that the lemma holds for $k$ for some $k \geq 1$.
In the paragraph preceding this lemma we have seen that $\tau_{k+1}(v) - \tau_k(v) \sim \Exp(k(n-k))$.
Thus, $\expected(\tau_{k+1}(v) - \tau_k(v)) = \frac 1{k(n-k)}$.
Plugging in the induction hypothesis yields
\begin{align*}
\expected\bigl(\tau_{k+1}(v)\bigr) & = \expected\bigl(\tau_{k}(v)\bigr) 
+ \frac{1}{k \cdot (n-k)}  = \frac 1n \cdot \left(H_{k-1} + H_{n-1} - H_{n-k} + \frac 1k + \frac{1}{n-k} \right)  \\
&= \frac 1n \cdot \bigl(H_{k} + H_{n-1} - H_{n-(k+1)}\bigr) .
\end{align*}
\end{proof}

From this result, we can easily deduce two known results:
averaging over $k$ yields that the expected distance of an edge
is $\frac{H_{n-1}}{n-1} \approx \ln n/n$~\cite{DavisPrieditis:ExpectedShortestPath:1993,Janson}.
The longest distance from $v$ to any other node is $\tau_n(v)$, which is $2 H_{n-1}/n \approx 2 \ln n/n$ in expectation~\cite{Janson}.
For completeness, let us mention that the diameter $\dmax$ is approximately $3 \ln n/n$~\cite{Janson}.
However, this does not follow immediately from Lemma~\ref{thm:expectations}.

\subsection[Distribution of tau k v]{\boldmath Distribution of $\tau_{k}(v)$}
\label{sec:Deltakv}

Let us now have a closer look at cumulative distribution function of $\tau_k(v)$ for fixed $v \in V$ and $k \in [n]$. 
To do this, the following lemma is very useful.

\begin{lemma}
\label{lem:sumequalsmax}
Let $X \sim \sum_{i=1}^n \Exp(ci)$.
Then $\probab(X \leq \alpha) = (1-e^{-c\alpha})^{n}$.
\end{lemma}

\begin{proof}
The random variable $X$ has the same distribution as $\max_{i=1}^n Y_i$, where $Y_i \sim \Exp(c)$.
We have $X \leq \alpha$ if and only if $Y_i \leq \alpha$ for all $i \in \{1, \ldots, n\}$.
\end{proof}

In the following, let $F_k$ denote the cumulative distribution function of $\tau_k(v)$ for some fixed vertex $v \in V$,
i.e., $F_k(x) = \probab(\tau_k(v) \leq x)$.

\begin{lemma}
\label{lem:cdf}
For every $\Delta \geq 0$, $v \in V$, and $k \in [n]$, we have
\[
  \bigl(1-\exp(-(n-k)\Delta)\bigr)^{k-1} \leq F_k(\Delta) \leq
   \bigl(1-\exp(-n\Delta)\bigr)^{k-1}.
\]
\end{lemma}

\begin{proof}
Lemma~\ref{thm:expectations} states that $\tau_{k}(v) \sim \sum_{i=1}^{k-1}\Exp(i(n-i))$.
We approximate the parameters
by $ci$ for $c \in \{n-k, n\}$. The distribution with $c = n$ is stochastically dominated
by the true distribution, which is in turn dominated by the distribution obtained for $c=n-k$.
We apply Lemma~\ref{lem:sumequalsmax} with $c = n$ and $c = n-k$.
\end{proof}

\begin{lemma}
\label{lem:cdf2}
Fix $\Delta \geq 0$ and a vertex $v \in V$.
Then
\[
\bigl(1-\exp(-(n-k)\Delta)\bigr)^{k-1} \leq \probab\bigl(|B_{\Delta}(v)| \geq k\bigr)
\leq \bigl(1-\exp(-n\Delta)\bigr)^{k-1}.
\]
\end{lemma}

\begin{proof}
We have $|B_{\Delta}(v)| \geq k$ if and only if $\tau_k(v) \leq \Delta$.
The lemma follows from Lemma~\ref{lem:cdf}.
\end{proof}

We can improve Lemma~\ref{lem:cdf} slightly in order to obtain even closer upper and lower bounds.
For $n, k \ge 2$, combining Lemmas~\ref{lem:cdf} and~\ref{lem:cdfImproved} yields tight upper and lower bounds
if we disregard the constants in the exponent, namely
$F_k(\Delta) = \big(1 - \exp(-\Theta(n \Delta))\big)^{\Theta(k)}$.

\begin{lemma} \label{lem:cdfImproved}
For all $v \in V$, $k \in [n]$, and $\Delta \geq 0$, we have
\[
 F_k(\Delta) \ge \bigl(1 - \exp(-(n-1) \Delta / 4)\bigr)^{n-1}
\]
and
\[
F_k(\Delta) \ge \bigl(1 - \exp(-(n-1) \Delta / 4)\bigr)^{\frac{4}{3}(k-1)}.
\]
\end{lemma}

\begin{proof}
As $\tau_{k}(v)$ is monotonically increasing in $k$, we have
$F_{k}(\Delta) \ge F_{k+1}(\Delta)$
    for all $k$. Thus, we have to prove the claim only for $k = n$.
    In this case, $\tau_{n}(v) \sim \sum_{i=1}^{n-1} \Exp(\lambda_{i})$,
    with $\lambda_{i} = i(n-i) = \lambda_{n-i}$.
    Setting $m = \lceil n/2 \rceil$ and exploiting the symmetry around $m$ yields
    \[
        \tau_{n}(v) \le \sum_{i=1}^{m} \Exp(\lambda_{i}) + \sum_{i=1}^{m} \Exp(\lambda_{i}) 
        = \tau_{m}(v) + \tau_{m}(v).
    \]
    Here, ``$\le$'' means stochastic dominance, ``$=$'' means equal distribution, 
    and ``$+$'' means adding up two independent random variables.
    Hence,
    \[
        F_{n}(\Delta) = \probab\bigl(\tau_{n}(v) \le \Delta\bigr)
       \ge \probab\bigl(\tau_{m}(v) + \tau_{m}(v) \le \Delta\bigr)
        \ge \probab\bigl(\tau_{m}(v) \le \Delta/2\bigr)^{2}.
    \]
    By Lemma~\ref{lem:cdf}, and using $m \le (n+1)/2$, this is bounded by
    \[
        F_{n}(\Delta) \ge (1-\exp(-(n-m)\Delta/2))^{2(m-1)}  
        \ge (1-\exp(-(n-1)\Delta/4))^{n-1}.
    \]
    
 For the second inequality, we use the first inequality of Lemma~\ref{lem:cdfImproved} for $k-1 \ge \frac{3}{4}(n-1)$ and
Lemma~\ref{lem:cdf} for $k-1 < \frac{3}{4}(n-1)$ as then $n-k \ge (n-1)/4$.
\end{proof}

\subsection[Tail Bounds for B Delta v and Delta max]{\boldmath Tail Bounds for $|B_{\Delta}(v)|$ and $\dmax$}

Our first tail bound for $|B_{\Delta}(v)|$, which is the number of vertices within distance $\Delta$
of a given vertex $v$, follows directly from Lemma~\ref{lem:cdf}.
From this lemma we derive the following corollary, which is a crucial ingredient for the existence of good
clusterings and, thus, for the analysis of heuristics in the remainder of this paper.

\begin{corollary}
\label{cor:tech}
Let $n \ge 5$ and fix $\Delta \geq 0$ and a vertex $v \in V$.
Then we have
\[
\probab\left(|B_{\Delta}(v)| < \min\left\{\exp\left(\Delta n/5\right), \frac{n+1}2 \right\}\right)
  \le \exp\left(-\Delta n/5\right).
\]
\end{corollary}

\begin{proof}
Lemma~\ref{lem:cdf2} yields
\begin{align*}
\probab\left(|B_{\Delta}(v)| < \min\left\{\exp\left(\Delta \frac{n-1}4\right), \frac{n+1}2 \right\}\right)
 & \leq 1- \left(1-\exp\left(-\frac{n-1}2 \Delta \right)\right)^{\exp(\Delta (n-1)/4)}  \\
 & \leq \exp\left(- \Delta \frac{n-1}4 \right),
\end{align*}
where the last inequality follows from $(1-x)^y \ge 1 - xy$ for $y \ge 1$, $x \ge 0$.
Using $(n-1)/4 \ge n/5$ for $n \ge 5$ completes the proof.
\end{proof}

Corollary~\ref{cor:tech} is almost tight according to the following result.

\begin{corollary}
\label{cor:techlower}
Fix $\Delta \ge 0$, a vertex $v \in V$, and any $c > 1$. Then
\[
\probab\bigl(|B_{\Delta}(v)| \geq \exp(c \Delta n)\bigr) < 
\exp\bigl(- (c-1) \Delta n\bigr).
\]
\end{corollary}
\begin{proof}
Lemma~\ref{lem:cdf2} with $k = c\Delta n$ yields
\[
\probab\bigl(|B_{\Delta}(v)| \geq \exp(c \Delta n)\bigr)
 \leq \bigl(1-\exp(-n\Delta)\bigr)^{\exp(c \Delta n)-1}.
\]
Using $1+x \le e^x$, we get
\[
\probab\bigl(|B_{\Delta}(v)| \geq \exp(c \Delta n)\bigr)
 \leq \exp\left(\exp(-n\Delta)- \exp\bigl((c-1) \cdot \Delta n\bigr)\right).
\]
Now, we bound $\exp(-n\Delta) \le 1$ and $\exp\bigl((c-1) \cdot \Delta n\bigr) \ge 1 + (c-1) \cdot \Delta n$, which yields the
claimed inequality.
\end{proof}

Janson~\cite{Janson} derived the following tail bound for the diameter $\dmax$. 
A qualitatively similar bound can be proved
using Lemma~\ref{lem:cdf2} and can also be derived from Hassin and Zemel's analysis~\cite{HassinZemel:ShortestPath:1985}.
However, Janson's bound is stronger with
respect to the constants in the exponent.

\begin{lemma}[\mbox{Janson~\cite[p. 352]{Janson}}]
\label{lem:dmax}
For any fixed $c > 3$, we have
$\probab(\dmax > c \ln (n)/n) \leq O(n^{3-c} \log^2 n)$.
\end{lemma}

\subsection{Balls and Clusters}
\label{ssec:clusters}

In this section,
we show our main structural contribution, which is a global property of random shortest path metrics.
We show that such instances can be divided into a small number of clusters of any given diameter.

From now on, let $s_\Delta = \min\{\exp(\Delta n/5), (n+1)/2\}$, as in Corollary~\ref{cor:tech}.
If $|B_{\Delta}(v)|$, the number of vertices within distance $\Delta$ of $v$, is
at least $s_\Delta$, then we call the vertex $v$ a \emph{dense $\Delta$-center},
and we call the set $B_\Delta(v)$ of vertices within distance $\Delta$ of $v$ (including $v$ itself)
the \emph{$\Delta$-ball of $v$}. Otherwise, if $|B_{\Delta}(v)| < s_\Delta$, and $v$ is not part of any $\Delta$-ball, we call the vertex $v$
a \emph{sparse $\Delta$-center}.
Any two vertices in the same $\Delta$-ball have a distance of at most~$2 \Delta$ because of the triangle inequality. 

If $\Delta$ is clear from the context, then we also speak about centers and balls without parameter. We can bound,
by Corollary~\ref{cor:tech},
the expected number of sparse $\Delta$-centers to be
at most $O(n/s_\Delta)$.

We want to partition the graph into a small number of clusters, each of diameter at most~$6
\Delta$. For this purpose, we put each sparse $\Delta$-center in its own cluster (of size 1). Then the
diameter of each such cluster is $0$, which is trivially upper-bounded by $6 \Delta$, and the number of these clusters is expected to be
at most $O(n/s_\Delta)$.

We are left with the dense $\Delta$-centers, which we cluster using the following algorithm:
Consider an auxiliary graph whose vertices are all dense $\Delta$-centers.
We draw an edge between two dense $\Delta$-centers $u$ and $v$ if $B_\Delta(u) \cap B_\Delta(v) \neq \emptyset$.
Now consider any maximal independent set of this auxiliary graph (for instance, a greedy independent set), and let $t$
be the number of its vertices.
Then we form initial clusters $C_1', \ldots, C_t'$, each containing one of the $\Delta$-balls corresponding
to the vertices in the independent set. By the independence, all these $t$ $\Delta$-balls are disjoint, which implies $t \le n/s_\Delta$.
The ball of every remaining center $v$ has at least one vertex in one of the $C_i'$.
We add all remaining vertices of $B_\Delta(v)$ to such a $C_i'$ to form the final clusters $C_1,\ldots,C_t$.
By construction, the diameter of each $C_i$ is at most $6\Delta$:
Consider any two vertices $u, v \in C_i$. The distance of $u$ towards its closest neighbor
in the initial ball $C_i'$ is at most $2\Delta$. The same holds
for $v$. Finally, the diameter of the initial ball $C_i'$ is also at most $2\Delta$.

With this partitioning, we have obtained the following structure:
We have an expected number of $O(n/s_\Delta)$ clusters of size 1 and diameter 0, and a number of $O(n/s_{\Delta})$ clusters
of size at least $s_\Delta$ and diameter at most $6 \Delta$.
Thus, we have $O(n/s_\Delta) = O(1 + n/\exp(\Delta n/5))$ clusters in total.
We summarize these findings in the following lemma.
This lemma is the crucial ingredient for bounding the expected approximation ratios of the greedy, nearest-neighbor,
and insertion heuristics.

\begin{lemma}
    \label{lem:numberofclusters}
Consider a random shortest path metric and let $\Delta \ge 0$. If we partition the instance into clusters, each of diameter at most $6 \Delta$,
then the expected number of clusters needed is $O(1 + n/\exp(\Delta n/5))$.
\end{lemma}

\section{Analysis of Heuristics}
\label{sec:heuristics}

\subsection{Greedy Heuristic for Minimum-Length Perfect Matching}

Finding min\-i\-mum-length perfect matchings in metric instances is the first problem that we consider.
This problem has been widely considered in the past and has applications in, e.g., optimizing the speed of
mechanical plotters~\cite{ReingoldTarjan,SupowitEA:HeuristicsMatching:1980}.
The worst-case running-time of $O(n^3)$ for finding an optimal matching is prohibitive
if the number $n$ of points is large.
Thus, simple heuristics are often used, with the greedy heuristic being probably the simplest one:
at every step, choose an edge of minimum length incident to the 
unmatched vertices and add it to the partial matching. Let $\greedy$ denote
the cost of the matching output by this greedy matching heuristic, and let $\mm$ denote the optimum value of the minimum-length perfect matching.
The worst-case approximation ratio for greedy matching on metric instances is $\Theta(n^{\log_2 (3/2)})$~\cite{ReingoldTarjan},
where $\log_2(3/2) \approx 0.58$.
In the case of Euclidean instances, the greedy algorithm has an approximation ratio of $O(1)$ with high probability
on random instances~\cite{ADS88}.
For independent random edge weights (without the triangle inequality), the expected weight of the matching
computed by the greedy algorithm is $\Theta(\log n)$~\cite{DyerEA:GreedyMatching:1993} whereas the optimal matching has a weight of $\Theta(1)$ with high probability,
which gives an $O(\log n)$ approximation ratio.

We show that greedy matching finds a matching of constant expected length on random shortest path metrics.

\begin{theorem}
    \label{thm:greedymatching}
     $\expected(\greedy)= O(1)$.
\end{theorem}

\begin{proof}
    Let $\Delta_i = \frac in$. We divide the run of $\greedy$ in phases as follows: we say that $\greedy$ is in
    phase $i$ if edges $\{u,v\}$ are inserted such that $d(u,v) \in (6\Delta_{i-1}, 6
    \Delta_i]$.
    Lemma~\ref{lem:dmax} allows to show that the expected sum of all edges longer than $\Delta_{\Omega(\log n)}$ is $o(1)$, so we can ignore them.

    $\greedy$ goes through phases $i$ with increasing $i$ (phases can be empty).
    We now estimate the contribution of phase $i$
    to the matching
    computed by $\greedy$.
    Using Lemma~\ref{lem:numberofclusters}, after phase $i-1$ we can find a clustering into clusters of diameter at most $6 \Delta_{i-1}$ using an expected number of $O(1 + n / e^{(i-1)/5})$ clusters. 
    Each such cluster can have at most one unmatched vertex.
    Thus, we have to add at most $O(1 + n /e^{(i-1)/5})$ edges in phase $i$. Each such edge connects vertices at a distance of at most $6 \Delta_i$.
    Hence, the contribution of phase $i$
    is $O(\frac in \cdot (1+ n/e^{(i-1)/5}))$ in expectation. Summing over all phases yields the desired bound:
    \[
		\expected\bigl(\greedy\bigr) =
        o(1) + \sum_{i=1}^{O(\log n)} O\left(\frac{i}{e^{(i-1)/5}} + \frac in \right)
        = O(1).
				\]
\end{proof}

Careful analysis allows us to bound the expected approximation ratio.

\begin{theorem}
    \label{thm:greedy-exp}
    The greedy algorithm for minimum-length perfect matching has constant approximation ratio on random shortest path metrics, i.e.,
    $\expected\left( \frac{\greedy}{\mm}\right) = O(1)$.
\end{theorem}

We will use the following tail bound to estimate the approximation ratios of the
greedy heuristic for matching as well as the
nearest-neighbor and insertion heuristics for the TSP.

\begin{lemma}
    \label{lem:opt-tail}
    Let $\alpha \in [0,1]$.
    Let $S_m$ be the sum of the lightest $m$ edge weights, where $m \geq \alpha n$.
    Then, for all $c \in [0,1]$, we have
    \[
      \probab(S_m \leq c) \leq \left(\frac{e^2 c}{2\alpha^2}\right)^{\alpha n}.
    \]

    Furthermore, $\tsp \geq \mm \geq S_{n/2}$, where $\tsp$ and $\mm$
    denote the length of the shortest TSP tour and the minimum-weight
    perfect matching, respectively, in the corresponding shortest path metric.
\end{lemma}

\begin{proof}
Let $X \sim \sum_{i=1}^m \Exp(1)$, and let $Y$ be the sum of $m$ independent random variables
drawn uniformly from $[0,1]$.
The random variable $X$ stochastically dominates $Y$, and $\probab(Y \leq c) = c^m/m!$.

The probability that $S_m \leq c$ is at most the probability that
there exists a subset of the edges of cardinality $m$ whose total weight is at most $c$.
By a union bound and using $\binom ab \leq (ae/b)^b$, $\binom n2 \leq n^2/2$, and $a! > (a/e)^a$, we obtain
\[
 \probab(S_m \leq c) \leq \binom{\binom n2}{m} \cdot \frac{c^m}{m!} \leq \left(\frac{n^2e^2c}{2m^2}\right)^m
 \leq \left(\frac{e^2c}{2\alpha^2}\right)^m.
\]
We can replace $m$ by its lower bound $\alpha n$ in the exponent~\cite[Fact 2.1]{kMeans} to obtain
the first claim.

It remains to prove $\tsp \geq \mm \geq S_{n/2}$.
The first inequality is trivial.
For the second inequality, consider a minimum-weight perfect matching in a random shortest path metric.
We replace every edge by the corresponding paths.
If we disregard multiple edges, then we are still left with at least $n/2$ edges
whose length is not shortened by taking shortest paths.
The sum of the weights of these $n/2$ edges is at most $\mm$ and at least $S_{n/2}$.
\end{proof}

\begin{proof}[Proof of Theorem~\ref{thm:greedy-exp}]
 The worst-case approximation ratio of $\greedy$ for minimum-weight perfect matching
 is $n^{\log_2(3/2)}$ \cite{ReingoldTarjan}.
 Let $c > 0$ be a sufficiently small constant.
 Then the approximation ratio of $\greedy$ on random shortest path instances
 is
 \[
\expected\left(  \frac{\greedy}{\mm}\right)
     \leq \expected\left( \frac{\greedy}c \right) + \probab(\mm < c) \cdot n^{\log_2(3/2)}.
  \]
   By Theorem~\ref{thm:greedymatching}, the first term is $O(1)$.
   Since $c$ is sufficiently small, Lemma~\ref{lem:opt-tail} shows that the
   second term is $o(1)$.
\end{proof}

\subsection{Nearest-Neighbor algorithm for the TSP}
\label{ssec:nn}

A greedy analogue for the traveling salesman problem (TSP) is the \emph{nearest neighbor}
heuristic: (1) Start with some starting vertex $v_0$ as the current vertex $v$. (2) At every iteration,
choose the nearest yet unvisited neighbor $u$ of the current vertex $v$ (called the successor of $v$) as the next vertex in the tour,
and move to the next iteration with the new vertex $u$ as the current vertex $v$. (3)
Go back to the first vertex $v_0$ if all vertices are visited.
Let $\nn$ denote both the nearest-neighbor heuristic itself
and the cost of the tour computed by it. Let $\tsp$ denote the cost of an optimal tour.
The nearest-neighbor heuristic $\nn$ achieves a worst-case ratio of $O(\log n)$ for metric
instances and also an average-case ratio (for independent, non-metric edge lengths) of $O(\log n)$~\cite{AusielloEA}.
We show that $\nn$ achieves a constant approximation ratio on random shortest
path instances. 

\begin{theorem}
\label{thm:nn-const}
For random shortest path instances we have
$\expected(\nn) = O(1)$
and
$\expected\left(\frac{\nn}{\tsp}\right) = O(1)$.
\end{theorem}

\begin{proof}
The proof is similar to the proof of Theorem~\ref{thm:greedy-exp}.
Let $\Delta_i = i/n$  for $i \in \mathbb{N}$. Let $Q= O(\log n/n)$ be sufficiently large.

Consider the clusters obtained with parameter $\Delta_i$ as in the discussion
preceding Lemma~\ref{lem:numberofclusters}. These clusters have diameters of at most $6\Delta_i$.
We refer to these clusters as the \emph{$i$-clusters}.
Let $v$ be any vertex.
We call $v$ \emph{bad at $i$}, if $v$ is in some $i$-cluster and $\nn$ chooses
a vertex at a distance of more than $6\Delta_i$ from $v$ for leaving $v$.
Hence, if $v$ is bad at $i$, then the next vertex lies outside 
of the cluster to which $v$ belongs.
(Note that $v$ is not bad at $i$ if the outgoing edge at $v$ leads to a neighbor
outside of the cluster of $v$ but at a distance of at most $6 \Delta_i$ from $v$.)

In the following, let the cost of a vertex $v$ be the distance from $v$ to its successor $u$.
The length of the tour produced by $\nn$ is equal to the sum of costs over all vertices.

\begin{claim}
\label{cla:nn}
The expected number of vertices with costs in the range
$(6 \Delta_{i}, 6\Delta_{i+1} ]$ is at most $O(1 + n/\exp(i/5))$.
\end{claim}

\begin{proof}[Proof of Claim~\ref{cla:nn}]
    Suppose that the cost of the neighbor chosen by $\nn$ for a vertex $v$ is in the interval $(6
    \Delta_{i}, 6\Delta_{i+1} ]$.  Then $v$ is bad at $i$. This happens only if all other vertices of the $i$-cluster containing $v$
    have already been visited. Otherwise, there would be another vertex $u$ in the same $i$-cluster with a distance of at most
    $6\Delta_{i}$ to $v$.
    By Lemma~\ref{lem:numberofclusters}, the number of $i$-clusters is at most $O(1 + n/\exp(i/5))$.
\end{proof}

If $\dmax \leq Q$, then it suffices to consider $i$ for $i \leq O(\log n)$.
    If $\dmax > Q$, then we bound the value of the tour produced by $\nn$ by $n \dmax$.
    This failure event, however, contributes only $o(1)$ to the expected value by Lemma~\ref{lem:dmax}.
For the case $\dmax \leq Q$, the contribution to the expected length of the $\nn$
tour is bounded from above by
\[
\sum_{i= 0}^{O(\log n)} 6 \Delta_{i+1} \cdot O\left(1 + \frac{n}{\exp(i/5)}\right) = \sum_{i= 0}^{O(\log n)} O\left(\frac{i+1}n + \frac{i+1}{\exp(i/5)}\right) = O(1).
\]

Using the fact that the worst-case approximation ratio of $\nn$ is $O(\log n)$, the proof
of the constant expected approximation ratio is similar to the proof of Theorem~\ref{thm:greedy-exp}.
\end{proof}

\subsection{Insertion Heuristics}
\label{ssec:insertion}

An insertion heuristic for the TSP is an algorithm that starts with an initial
tour on a few vertices and extends this tour iteratively by adding the
remaining vertices.
In every iteration, a vertex is chosen according to some rule,
and this vertex is inserted at the place in the current tour where it increases
the total tour length the least.
The approximation ratio achieved depends on the rule used for selecting the next node to insert.
Certain insertion heuristics such as nearest neighbor insertion (which is different from the nearest neighbor algorithm from
the previous section)
achieve constant approximation ratios~\cite{Rosenkrantz77}. The random insertion algorithm,
where the next vertex is chosen uniformly at random from the remaining vertices, has a
worst-case approximation ratio of $\Omega(\log\log n/\log\log \log n)$, and 
there are insertion heuristics with a worst-case approximation ratio of
$\Omega(\log n/ \log\log n)$ ~\cite{Azar:InsertionLowerTSP:1994}.

A rule $R$ that specifies an insertion heuristic can be viewed as follows:
depending on the distances $d$, it (1) chooses a set $R_V$ of vertices for computing an initial tour
and (2) given any tour of vertices $V' \supseteq R_V$, describes how to choose the next vertex.
Let $\instsp_R$ denote the length of the tour produced with rule $R$.

For random shortest path metrics, we show that any insertion heuristic produces a tour
whose length is expected to be within a constant factor of the optimal tour.
This result holds irrespective of which insertion strategy we actually use.

\begin{theorem}
    \label{thm:tsp-ins}
    For every rule $R$, we have $\expected(\instsp_R) = O(1)$
    and $\expected\bigl(\frac{\instsp_R}{\tsp}\bigr) = O(1)$.
\end{theorem}
\begin{proof}
   Let $\Delta_i = i/n$ for $i \in \mathbb{N}$ and $Q=O(\log n/n)$ be sufficiently large. Assume that $\dmax \leq Q$.
    If $\dmax > Q$, then we bound the length of the tour produced by $n \cdot \dmax$.
    This contributes only $o(1)$ to the expected value of length of the tour produced by Lemma~\ref{lem:dmax}.

    Suppose we have a partial tour $T$ and
    $v$ is the vertex that we have to insert next. If $T$ has a vertex $u$ such that $v$ and $u$ are
    in a common $i$-cluster, then the triangle inequality implies that the
    costs of inserting $v$ into $T$ is at most $12 \Delta_i$ because the diameters of $i$-clusters
    are at most $6 \Delta_i$~\cite[Lemma 2]{Rosenkrantz77}.
    For each $i$, only the insertion of the first vertex of each $i$-cluster can possibly cost more
    than $12\Delta_i$.
    Thus, the number of vertices whose insertion would incur costs in the range
    $(12\Delta_{i}, 12\Delta_{i+1}]$ is at most $O\bigl(1 + \frac{n}{\exp(i/5)}\bigr)$
    in expectation. Note that we only have to consider $i$ with $i \le O(\log n)$ since $\dmax \le Q$.
    The expected costs of the initial tour are at most $\tsp = O(1)$~\cite{Frieze04}.
    Summing up the expected costs for all $i$ plus the costs of the initial tour, we obtain that the expected costs of the tour obtained by
    an insertion heuristic is bounded from above by
    \begin{align*} 
        \expected(\instsp_R) = O(1) + \sum_{i=0}^{O(\log n)} \Delta_i \cdot O\left(1 + \frac{n}{\exp(i/5)}\right)
        &= O(1).
    \end{align*}
    Note that the above argument is independent of the rule $R$ used.
    
    The proof for the approximation ratio is similar to the proof of Theorem~\ref{thm:greedy-exp} and uses the worst-case ratio of $O(\log n)$
    for insertion heuristics for any rule $R$~\cite[Theorem 3]{Rosenkrantz77}.
\end{proof}

\subsection{Running-Time of 2-Opt for the TSP}
\label{ssec:2opt}

The 2-opt heuristic for the TSP starts with an initial tour and successively improves the
tour by so-called 2-exchanges until no further refinement is possible.
In a 2-exchange, a pair of edges $e_{12}=\{v_1,v_2\}$ and $e_{34}=\{v_3,v_4\}$, where $v_1, v_2, v_3, v_4$ appear in this order in the
Hamiltonian tour, are replaced by a pair of edges $e_{13} =\{v_1,v_3\}$ and $e_{24} = \{v_2,v_4\}$ to get a shorter tour.
The 2-opt heuristic is easy to implement and widely used. In practice, it usually converges quite quickly to close-to-optimal
solutions~\cite{TSPExperiments}.
To explain its performance in practice, probabilistic analyses of its running-time on geometric instances~\cite{Englert:2opt,TwoOptGauss,Switching}
and its approximation performance on geometric instances~\cite{Englert:2opt} and with independent, non-metric edge lengths~\cite{EM2opt}
have been conducted.
We prove that for random shortest path metrics, the expected number of iterations that 2-opt needs is bounded by a polynomial.

\begin{theorem}
\label{thm:2-opt}
The expected number of iterations that 2-opt needs to find a local optimum is bounded by $O(n^8\log^3 n)$.
\end{theorem}

\begin{proof}
The proof is similar to the analysis of 2-opt by Englert et al.~\cite{Englert:2opt}. Consider a 2-exchange where edges $e_1$ and $e_2$ are replaced by
edges $f_1$ and $f_2$ as described above. The improvement obtained from this
exchange is given by
$\delta = \delta(v_1, v_2, v_3, v_4)= d(v_1,v_2)+d(v_3,v_4)- d(v_1, v_3)-d(v_2, v_4)$.

We estimate the probability $\probab(\delta \in (0, \eps])$ of the event
that the improvement is at most $\eps$ for some $\eps > 0$.
The distances $d(v_i, v_j)$ correspond to shortest paths with respect to the exponentially distributed edge weights
$w$. Assume for the moment that we know these
paths. Then we can rewrite the improvement as
\begin{equation}
\delta = \sum_{e \in E} \alpha_e w(e) \label{equ:delta}
\end{equation}
for some coefficients $\alpha_e \in \{-2, -1, 0,1,2\}$.
If the exchange considered is indeed a 2-exchange, then $\delta> 0$. Thus, in this case, there exists at least on edge $e = \{u, u'\}$ with $\alpha_e \neq 0$.
Let $I \subseteq \{e_{12}, e_{34}, e_{13}, e_{24}\}$ be the set of edges of the 2-exchange such that the corresponding paths use~$e$.

For all combinations of $I$ and $e$, let $\delta^{I,e}_{ij}$ be the following quantity:
\begin{itemize}
\item If $e_{ij} \notin I$, then $\delta^{I,e}_{ij}$ is the length of a shortest path
from $v_i$ to $v_j$ without using $e$.
\item If $e_{ij} \in I$, then $\delta^{I, e}_{ij}$ is the minimum of
\begin{itemize}
\item the length of a shortest path from $v_i$ to $u$ without $e$ plus the length of a shortest path from $u'$ to $v_j$ without $e$ and
\item the length of a shortest path from $v_i$ to $u'$ without $e$ plus the length of a shortest path from $u$ to $v_j$ without $e$.
\end{itemize}
\end{itemize}
Let $\delta^{e,I} = \delta^{e,I}_{12} + \delta^{e,I}_{34} - \delta^{e,I}_{13} - \delta^{e,I}_{24}$.

\begin{claim}
\label{existsEI}
For every outcome of the random edge weights,
there exists an edge $e$ and a set $I$ such that $\delta = \delta^{e,I} + \alpha w(e)$, where $\alpha \in \{-2, -1, 1, 2\}$ is determined by $e$ and $I$.
\end{claim}

\begin{proof}[Proof of Claim~\ref{existsEI}]
Fix the edge weights arbitrarily and consider any four shortest paths. Then there exists some edge $e$ with non-zero $\alpha_e$ in \eqref{equ:delta}.
We choose this $e$, an appropriate set $I$, and we choose $\alpha = \alpha_e$. Then the claim follows from the definition of $\delta^{e,I}$.
\end{proof}

Claim~\ref{existsEI} yields that $\delta \in (0,\eps]$ implies that there are an $e$ and an $I$ with
$\delta^{e,I} + \alpha w(e) \in (0,\eps]$.

\begin{claim}
\label{cla:eps}
Let $e$ and $I$ be arbitrary with $\alpha = \alpha_e \neq 0$. Then $\probab(\delta^{e,I} + \alpha w(e) \in (0,\eps]) \leq \eps$.
\end{claim}

\begin{proof}[Proof of Claim~\ref{cla:eps}]
We fix the edge weights of all edges except for $e$.
This determines $\delta^{e,I}$. Thus, $\delta^{e,I} + \alpha w(e) \in (0,\eps]$ if and only
of $w(e)$ assumes a value in a now fixed interval of size $\eps/\alpha \leq \eps$. Since the density of the exponential distribution is bounded from above by $1$,
the claim follows.
\end{proof}

The number of possible choices for $e$ and $I$ is $O(n^2)$. Thus, $\probab(\delta \in (0,\eps]) = O(n^2\eps)$.

Let $\deltamin > 0$ be the minimum improvement made by any 2-exchange. Since there are at most $n^4$ different 2-exchanges, we have
$\probab(\deltamin \leq \eps) = O(n^6 \eps)$.

The initial tour has a length of at most $n \dmax$. Let $T$ be the number of iterations that 2-opt takes.
Then $T \leq n \dmax/\deltamin$.
Now, $T > x$ implies $\dmax/\deltamin > x/n$.  The event  $\dmax/\deltamin>x/n$ is contained in the union of the events  $\dmax > \log x \ln n/n$, and  $\deltamin < \ln n \cdot \log x/x$.
The first happens with a probability of at most $n^{-\Omega(\log(x))}$ by Lemma~\ref{lem:dmax}.
The second happens with a probability of at most $O(n^6 \log(x)/x)$.
Thus, we obtain
\[
  \probab(T>x) \leq n^{-\Omega(\log(x))} + O\bigl(n^6 \ln n \cdot \log(x)/x\bigr).
\]
Since the number of iterations is at most $n!$, we obtain
an upper bound of
\[
  \expected(T) \leq \sum_{x=1}^{n!} \left( n^{-\Omega(\log(x))} + O(n^6 \ln n \log(x)/x) \right).
\]
The sum of the $n^{-\Omega(\log(x))}$ is negligible. The sum of the $O(n^6 \ln n \log(x)/x)$
contributes $O(n^6 \ln n \log(n!)^2) = O(n^8 \log^3 n)$.
\end{proof}

\section[k-Median]{\boldmath $k$-Median}
\label{sec:kmedian}

In the (metric) $k$-median problem, we are given a finite metric space $(V,d)$ and should pick $k$ points $U \subseteq V$ such that $\sum_{v \in V} \min_{u \in U} d(v,u)$ is minimized. We call the set $U$ a $k$-median.
Regarding worst-case analysis, the best known approximation algorithm for this problem achieves an approximation ratio of $3+\eps$~\cite{AryakMedian}.

In this section, we consider the $k$-median problem in the setting of random shortest path metrics. In particular we examine the approximation ratio of the algorithm $\trivial$, which picks $k$ points independently of the metric space, e.g., $U = \{1,\ldots,k\}$ or $k$ random points in~$V$. We show that $\trivial$ yields a $(1+o(1))$-approximation for $k = O(n^{1-\eps})$. This can be seen as an algorithmic result since it improves upon the worst-case approximation ratio, but it is essentially a structural result on random shortest path metrics. It means that any set of $k$ points is, with high probability, a very good $k$-median, which gives some knowledge about the topology of random shortest path metrics. For larger, but not too large $k$, i.e., $k \le (1-\eps)n$, $\trivial$ still yields an $O(1)$-approximation. 

The main insight comes from generalizing the growth process described in Section~\ref{sec:Deltakv}.
Fixing $U = \{v_1,\ldots,v_k\} \subseteq V$ we sort the vertices $V \setminus U$ by their distance to $U$ in ascending order, calling the resulting order $v_{k+1},\ldots, v_n$. Now we consider $\delta_i = d(v_{i+1},U) - d(v_i,U)$ for $k \le i < n$. These random variables are generated by a simple growth process analogous to
the one described in Section~\ref{sec:Deltakv}.
This shows that the $\delta_i$ are independent and $\delta_i \sim \Exp(i \cdot (n-i))$.
Since
$a \Exp(b) \sim \Exp(b/a)$, we have
\[
\cost(U) = \sum_{i=k}^{n-1} (n-i) \cdot \delta_i \sim \sum_{i=k}^{n-1} (n-i) \cdot \Exp(i \cdot (n-i)) 
    \sim \sum_{i=k}^{n-1} \Exp(i).
\]
From this, we can read off the expected cost of $U$ immediately, and thus the expected cost of $\trivial$.

\begin{lemma} \label{lem:ExTrivial}
    Fix $U \subseteq V$ of size $k$. We have
    \[
        \expected(\trivial) = \expected\bigl(\cost(U)\bigr) = H_{n-1} - H_{k-1} = \ln( n/k ) + \Theta(1).
    \] 
\end{lemma}

\begin{proof}
    We have $\expected(\cost(U)) = \sum_{i=k}^{n-1} \frac 1i = H_{n-1} - H_{k-1}$. Using $H_n = \ln(n) + \Theta(1)$ yields the last equality.
\end{proof}

By closely examining the random variable $\sum_{i=k}^{n-1} \Exp(i)$, we can show good tail bounds for the probability
that the cost of $U$ is lower than expected. Together with the union bound this yields tail bounds for the optimal $k$-median $\kmedian$,
which implies the following theorem.
In this theorem, the approximation ratio becomes $1 + O\big(\frac{\ln \ln(n)}{\ln(n)}\big)$
for $k = O(n^{1-\eps})$.

\begin{theorem}
\label{thm:kcentre}
    Let $k \le (1-\eps)n$ for some constant $\eps > 0$. 
    Then 
    \[
        \expected\left( \frac{\trivial}{\kmedian} \right) = O(1).
    \]
    If we have $k \le \kappa n$ for some sufficiently small constant $\kappa \in (0,1)$, then
    \begin{equation}
       \expected\left( \frac{\trivial}{\kmedian} \right) = 1 + O\left( \frac{\ln \ln (n/k)}{\ln (n/k)} \right).
       \label{kcenterthm}
    \end{equation}
\end{theorem}

We need the following lemmas to prove Theorem~\ref{thm:kcentre}.

\begin{lemma}
\label{densitycenter}
The density $f$ of $\sum_{i=k}^{m} \Exp(i)$ is given by
\[
  f(x) = k \cdot \binom{m}k \cdot \exp(-kx) \cdot \bigl(1-\exp(-x)\bigr)^{m-k}.
\]
\end{lemma}

\begin{proof}
The distribution $\sum_{i = k}^{m} \Exp(i)$ corresponds to the $k$-th largest element of a set of
$m$ independent, exponentially distributed random variables with parameter $1$. The density
of such order statistics is known~\cite[Example 2.38]{Ross}.
\end{proof}

\begin{lemma} \label{lem:kcentTail}
    Let $c > 0$ be sufficiently large, and let $k \le c' n$ for $c'=c'(c)>0$ be sufficiently small. Then 
    \[
        \probab\left(\kmedian < \ln\left( \frac nk \right) - \ln \ln\left( \frac nk \right) - \ln c \right) = n^{-\Omega(c)}.
    \]
\end{lemma}

\begin{proof}
    Fix $U \subseteq V$ of size $k$ and consider $\cost(U) \sim \sum_{i=k}^{n-1} \Exp(i)$.
    In the following we set $m := n-1$ to shorten notation. We now want to bound $f(x)$ from above
    at $x = \ln\bigl( \frac m{a k} \bigr)$ for a sufficiently large $a$ with $1 \le a \le m/k$
    (such an $a$ exists since $k$ is small enough).
    Plugging in this particular $x$ and using $\binom m  k \le m^k e^k / k^k$ yields
    \[
        f(x) 
         = k \cdot \binom m k  \cdot \frac{a^k k^k (m-ak)^{m-k} }{m^m} 
         \le k (e a)^k \left( 1  -\frac{ak}m \right)^{m-k}.
    \]
    Using $1+x \le e^x$ and $m-k = \Omega(m)$, so that $(m-k)/m = \Omega(1)$, yields
    \begin{align*}
        f(x) 
        & \le k (e a)^k \exp(-\Omega(a k)).
    \end{align*}
    Since $a$ is sufficiently large, the first two factors are lower order terms that we can hide by the $\Omega$.
    Thus, we can simplify this further to
    \begin{align*}
        f(x) 
        & \le \exp(-\Omega(a k)).
    \end{align*}
    Rearranging this using $a = \frac mk e^{-x}$ yields
    \begin{equation}
        f(x) = \exp(-\Omega( m \exp(-x)), \label{kcenterdensityapprox}
    \end{equation}
    which holds for any $x \in [0,\ln\bigl( \frac m{\alpha k} \bigr)]$ for any sufficiently large $\alpha \ge 1$.
    Now we can bound the probability that $\cost(U) < \ln\bigl( \frac m{\alpha k} \bigr)$. This probability is equal to
    \begin{align*}
        \int_0^{\ln( \frac m{\alpha k} )} f(x) \, \text dx
        & = \int_0^{\ln( \frac m{\alpha k} )} f\left(\ln\left( \frac m{\alpha k} \right) - x\right) \, \text dx  \\
        & = \int_0^{\ln( \frac m{\alpha k} )} \exp\bigl(-\Omega( \alpha k \exp(x) )\bigr) \, \text dx  & \text{using \eqref{kcenterdensityapprox}}\\
        & \le \int_0^{\infty} \exp\bigl(-\Omega( \alpha k (1+x) )\bigr) \, \text dx  
         \le \exp\bigl(-\Omega(\alpha k)\bigr)
    \end{align*}
    since $\int_0^{\infty} \exp(-\Omega( \alpha k x )) \, \text dx = O(1/(\alpha k)) \le 1$ as $\alpha$ is 
    sufficiently large.
    
    In order for $\kmedian$ to be less than $\ln\bigl( \frac m{\alpha k} \bigr)$, one of the subsets $U \subseteq V$ of size $k$ has to have cost less than $\ln\bigl( \frac m{\alpha k} \bigr)$. We bound the probability of the latter using the union bound and get
    \begin{align*}
        \probab\left(\kmedian < \ln\left( \frac m{\alpha k} \right)\right)
        & = \probab\left(\exists U \subseteq V, |U| = k \colon \cost(U) < \ln\left( \frac m{\alpha k} \right)\right) \\
        & \le 
        \binom n k \cdot \probab\left(\cost(U) < \ln\left( \frac m{\alpha k} \right)\right) \\
        & \le \binom n k \cdot \exp\bigl(-\Omega(\alpha k)\bigr).
    \end{align*}
    By setting $\alpha = c \ln\bigl(\frac nk \bigr)$ for sufficiently large $c \ge 1$, we fulfill all conditions on $\alpha$. This yields
    \begin{align*}
        \probab\left(\kmedian < \ln\left( \frac nk \right) - \ln \ln\left( \frac nk \right) - \ln c\right)
        & \le \left( \frac{e n}{k} \right)^k \cdot \left(\frac nk \right)^{-\Omega(c k)}.
    \end{align*}
    Since $k$ is sufficiently smaller than $n$, we have $\frac{e n}{k} \le (\frac nk)^2$.
    Thus, for sufficiently large $c$, the right hand side simplifies to $( \frac nk )^{-\Omega(c k)}$. 
    Since $k$ is at least 1 and sufficiently smaller than $n$, we have $(\frac nk)^k \ge n$. Thus, the probability is bounded by $n^{-\Omega(c)}$, which finishes the proof.
\end{proof}

To bound the expected value of the quotient $\trivial/\kmedian$,
we further need to bound the probabilities that $\trivial$ is much too large or $\kmedian$ is much too small. This is achieved by the following two lemmas.

\begin{lemma} \label{lem:medianconst}
    Let $k \le (1-\eps)n$ for some constant $\eps > 0$. 
    Then, for any $c > 0$, we have
    \[
      \probab(\kmedian < c) = O(c)^{\Omega(n)}.
    \]
\end{lemma}

\begin{proof}
   Since $n-k$ vertices have to be connected to the $k$-median,
    the cost of the $k$-median is the sum of $n-k$ shortest path lengths. 
    Thus, the cost of the minimal $k$-median is at least the sum of the smallest $n-k$ edge weights $w(e)$.
    We use Lemma~\ref{lem:opt-tail} with $\alpha = \eps$.
\end{proof}

\begin{lemma} \label{lem:TrivUpper}
    For any $c \ge 3$, we have 
    $\probab(\trivial > n^c) \le \exp(-n^{c/3})$.
\end{lemma}

\begin{proof}
    We can bound very roughly
    $\trivial \le n \max_e\{w(e)\}$. As $\max_e\{w(e)\}$ is the maximum of $\binom n2$ independent exponentially distributed random variables,
    we have
    \begin{align*}
        \probab\bigl(\trivial \le n^c\bigr)&
        \ge (1-\exp(- n^{c-1}))^{\binom n2} \geq 1 - \binom n2 \cdot \exp(-n^{c-1}) \\
        & \geq 1 - \exp\bigl(-n^{c-2}\bigr) \geq 1 - \exp\bigl(-n^{c/3}\bigr).
    \end{align*}
\end{proof}

\begin{proof}[Proof of Theorem~\ref{thm:kcentre}]
    Let $T = \trivial$ and $C = \kmedian$ for short. We have for any $m \ge 0$
    \begin{align} \label{eq:TCsplit}
        \expected\left( \frac TC \right) \le \expected\left( \frac Tm \right) + \probab(C<m) \cdot \expected\left(\frac TC \bigmid C < m \right).
    \end{align}

\noindent
    \emph{Case 1 ($k \leq c' n$, $c'$ sufficiently small):} Using Lemma~\ref{lem:kcentTail},
    we can pick $c > 0$ such that 
    \begin{align*}
        \probab\left[C < \ln\left( \frac nk \right) - \ln \ln\left( \frac nk \right) - \ln c \right] \le n^{-7}.
    \end{align*}
    Set $m = \ln\left( \frac nk \right) - \ln \ln\left( \frac nk \right) - \ln c$. 
    Then, by Lemma~\ref{lem:ExTrivial}, we have
    \begin{align*}
        \expected\left( \frac Tm \right) \le \frac{ \ln(n/k) + O(1) }m \le 1 + O\left( \frac{\ln \ln(n/k)}{\ln(n/k)}\right).
    \end{align*}
    We show that the second summand of inequality~\eqref{eq:TCsplit} is $O(1/n)$ 
    in the current situation, which shows the claim.
    We have
    \begin{align*}
        \probab(C<m) \cdot \expected\left( \frac TC \bigmid C < m \right)
        &= \probab(C<m) \cdot \int_0^{\infty} \probab\left(\frac TC \ge x \bigmid C < m\right) \, \text dx  \\
        &\le \probab(C<m) \cdot \left( n^6 + \int_{n^6}^{\infty} \probab\left(\frac TC \ge x \bigmid C < m\right) \, \text dx \right)  \\
        &\le n^{-1} + \int_{n^6}^{\infty} \probab\left( \frac TC \ge x \text{ and } C < m\right) \, \text dx  \\
        &\le n^{-1} + \int_{n^6}^{\infty} \probab\left( \frac TC \ge x\right) \, \text dx  \\
        &\le n^{-1} + \int_{n^6}^{\infty} 2\max\left\{\probab\left(T \ge \sqrt{x}\right), \probab\left(C \le \frac{1}{\sqrt{x}}\right)\right\} \, \text dx
    \end{align*}
    since $T/C \ge x$ implies $T \ge \sqrt{x}$ or $C \le 1/\sqrt{x}$.
    Using Lemmas~\ref{lem:medianconst} and~\ref{lem:TrivUpper}, this yields
    \[
        \probab(C<m) \cdot \expected\left( \frac TC \bigmid C < m \right)
        \le n^{-1} + \int_{n^6}^{\infty} 2\max\left\{ \exp\bigl(-x^{1/6}\bigr), O\left(\frac 1{\sqrt{x}}\right)^{\Omega(n)} \right\} \, \text dx 
         = O(1/n).
    \]
    
\noindent
\emph{Case 2 ($c' n < k \le (1-\eps)n$):} We repeat the proof above, now choosing $m$ to be a
    sufficiently small constant. Then $\probab(C < m) = O(m)^{\Omega(n)} \le O(n^{-7})$
    by Lemma~\ref{lem:medianconst}, and we have
    \begin{align*}
        \expected\left( \frac Tm \right) = \frac{\ln(n/k) + O(1)}m = O(1),
    \end{align*}
    since $k > c' n$. Together with the first case, this shows the first claim.
\end{proof}

\section{Concluding Remarks}
\label{sec:open}

\subsection{General Probability Distributions}

Using a coupling argument, Janson~\cite[Section 3]{Janson} proved 
that the results about the length of a fixed edge and the longest edge carry over if the exponential distribution
is replaced by a probability distribution with the following property: the probability that an edge weight is smaller than
$x$ is $x+o(x)$. This property is satisfied, e.g., by the exponential distribution with parameter 1 and by the uniform distribution
on the interval $[0,1]$.
The intuition is that, because the longest edge has a length of $O(\log n/n) = o(1)$,
only the behavior of the distribution in a small, shrinking interval $[0, o(1)]$ is relevant and the $o(x)$ term becomes irrelevant.

We believe that also all of our results carry over to such probability distributions. In fact, we started our research using the uniform distribution
and only switched to exponential distributions because they are technically easier to handle.
However, we decided not to carry out the corresponding proofs because, first, they seem to be technically very tedious and, second,
we feel that they do not add much.

\subsection{Open Problems}

To conclude the paper, let us list the open problems that we consider most interesting:
\begin{enumerate}
\item While the distribution of distances in asymmetric instances does not differ much from the symmetric case, an obstacle in the
   application of asymmetric random shortest path metrics seems to be the lack of clusters of small diameter (see Section~\ref{sec:structure}).
Is there an asymmetric counterpart for this?
\item Is it possible to prove an $1+o(1)$ approximation ratio (like Dyer and Frieze~\cite{DyerFrieze} for the patching algorithm)
   for any of the simple heuristics that we analyzed?
\item What is the approximation ratio of 2-opt in random shortest path metrics?
In the worst case on metric instances, it is $O(\sqrt n)$~\cite{CKT99}.
For independent, non-metric edge lengths drawn uniformly from the interval $[0,1]$, the expected approximation ratio is
$O(\sqrt n \cdot \log^{3/2} n)$~\cite{EM2opt}.
For $d$-dimensional geometric instances, the smoothed approximation ratio is
$O(\phi^{1/d})$~\cite{Englert:2opt}, where $\phi$ is the perturbation parameter.

We easily get an approximation ratio of $O(\log n)$ based
on the two facts that the length of the optimal tour is $\Theta(1)$ with high probability and that $\dmax = O(\log n/n)$
with high probability. Can we prove that the expected ratio of 2-opt is
$o(\log n)$?
\end{enumerate}


\end{document}